\def\year{2019}\relax
\documentclass[letterpaper]{article} 
\usepackage{aaai19}  
\usepackage{times}  
\usepackage{helvet}  
\usepackage{courier}  
\usepackage{url}  
\usepackage{graphicx}  
\usepackage{dblfloatfix}
\usepackage[utf8]{inputenc} 
\usepackage[T1]{fontenc}    
\usepackage{booktabs}       
\usepackage{amsfonts}       
\usepackage{nicefrac}       
\usepackage{microtype}      
\usepackage{times}
\usepackage{amsmath}
\usepackage{amsthm}
\usepackage{amssymb}
\usepackage{subfig} 
\usepackage{bm} 
\usepackage{mdframed}
\usepackage{multirow}
\usepackage{mathtools}
\usepackage{times}
\usepackage{xcolor}
\usepackage{soul}
\usepackage[linesnumbered,ruled,norelsize]{algorithm2e}
\usepackage{algorithmic}

\usepackage{natbib}
\usepackage{url}
\usepackage[colorlinks=true,bookmarks=true, pdftex, citecolor = black, linkcolor = black, urlcolor = black]{hyperref}
\hypersetup{%
	breaklinks,
	baseurl       = http://,
	pdfborder     = 0 0 0,
	pdfpagemode   = UseNone,
	pdfpagelabels = false,
	pdfstartpage  = 1,
	bookmarksopen = true,
	bookmarksdepth= 3,
}

\DeclareMathOperator*{\argmin}{arg\,min}
\DeclareMathOperator*{\argmax}{arg\,max}

\newcommand{\ds}{\displaystyle}

\newcommand{\F}{\mathcal{F}}

\newcommand{\X}{\mathcal{X}}

\newcommand{\Reals}{\mathbb{R}}
\newcommand{\cond}{\:|\:}

\newcommand{\ones}{\mathbf{1}}

\definecolor{gray}{gray}{0.50}

\newtheorem{proposition}{Proposition}

\newtheorem{definition}{Definition}

\begin{document}
 	%
 	\title{Fast Iterative Combinatorial Auctions via Bayesian Learning}
 	\author{
 		Gianluca Brero\\
 		University of Zurich\\
 		{\small\tt brero@ifi.uzh.ch}
 		\And
 		S\'{e}bastien Lahaie\\
 		Google Research\\
 		{\small\tt slahaie@google.com}
 		\And
 		Sven Seuken\\
 		University of Zurich\\
 		{\small\tt seuken@ifi.uzh.ch}
 	}
 	\maketitle
 	
 	\begin{abstract}
 		Iterative combinatorial auctions (CAs) are often used in multi-billion dollar domains like spectrum auctions, and speed of convergence is one of the crucial factors behind the choice of a specific design for practical applications. To achieve fast convergence, current CAs require careful tuning of the price update rule to balance convergence speed and allocative efficiency. \citet{brero2018Bayes} recently introduced a Bayesian iterative auction design for settings with single-minded bidders. The Bayesian approach allowed them to incorporate prior knowledge into the price update algorithm, reducing the number of rounds to convergence with minimal parameter tuning. In this paper, we generalize their work to settings with no restrictions on bidder valuations. We introduce a new Bayesian CA design for this general setting which uses Monte Carlo Expectation Maximization to update prices at each round of the auction.
 		We evaluate our approach via simulations on CATS instances. Our results show that our Bayesian CA outperforms even a highly optimized benchmark in terms of clearing percentage and convergence speed.
 	\end{abstract}
 	
 	
 	\section{Introduction}
 	\label{sec:introduction}
 	In a combinatorial auction (CA), a seller puts multiple indivisible items up for sale among several buyers who place bids on packages of items. By placing multiple package bids, a buyer can express complex preferences where items are complements, substitutes, or both. CAs have found widespread applications, including for spectrum license allocation~\citep{Cramton2013SpectrumAuctionDesign}, the allocation of TV advertising slots \citep{Goetzendorf:2015}, and industrial procurement \citep{Sandholm2013}.
 	
 	Practical auctions often employ \emph{iterative} designs, giving rise to \emph{iterative combinatorial auctions}, where bidders interact with the auctioneer over the course of multiple rounds. A well-known example is the combinatorial clock auction (CCA) which has been used by many governments around the world to conduct their spectrum auctions, and it has generated more than \$20 Billion in revenue since 2008 \citep{ausubel2017practical}. The CCA consists of two phases: an initial clock phase used for price discovery, followed by a sealed-bid phase where bidders can place additional bids.
 	
 	One of the key desiderata for a CA is its speed of convergence because each round can involve costly computations and business modeling on the part of the bidders \citep{kwasnica2005new,milgrom2013designingUS,bichler2017coalition}. Large spectrum auctions, for example, can easily take more than 100 bidding rounds.\footnote{See for example: \url{https://www.ic.gc.ca/eic/site/smt-gst.nsf/eng/sf11085.html}} To lower the number of rounds, in practice many CAs use aggressive price updates (e.g., increasing prices by 5\% to 10\% every round), which can lead to low allocative efficiency~\citep{ausubel2017practical}. Thus, the design of iterative CAs that are highly efficient but also converge in a small number of rounds still remains a challenging problem.

 	\subsection{Machine Learning in Auction Design}
 	AI researchers have studied this problem from multiple angles. One early research direction has been to employ machine learning (ML) techniques in preference elicitation~\citep{lahaie2004applying,blum2004preference}. In a related thread of research, \citeauthor{brero2017probably} (\citeyear{brero2017probably}, \citeyear{brero2018MLelicit}) integrated ML into a CA design, but they used value queries instead of demand queries (prices).
 	
 	In recent work, \citet{brero2018Bayes} proposed a Bayesian price-based iterative CA that integrates prior knowledge over bidders' valuations to achieve fast convergence and high allocative efficiency. In their design, the auctioneer maintains a model of the buyers' valuations which is updated and refined as buyers bid in the auction and reveal information about their values. The valuation model is used to compute prices at each round to drive the bidding process forward. However, their design has two substantial limitations: (i) it only works with single-minded bidders (i.e., each bidder is only interested in one bundle), and (ii) it only allows for Gaussian models of bidder valuations (even if other models are more suitable and accurate). These limitations are fundamental, because their design relies on both assumptions to obtain an analytical form for the price update rule.
 	
 	Similarly to \citet{brero2018Bayes}, \citeauthor{nguyen2014optimizing}  (\citeyear{nguyen2014optimizing}, \citeyear{nguyen2016multi}) studied different ways to determine prices in reverse auctions based on probabilistic knowledge on bidders' values. However, in contrast to the setting studied by \citet{brero2018Bayes}, in these papers bidders' valuations were not combinatorial, and the auctioneer was allowed to propose personalized prices to each bidder.
 	
 	\subsection{Overview of our Approach}
 	In this paper, we generalize the approach by~\citet{brero2018Bayes}. We propose a new, general Bayesian CA that can make use of any model of bidders' valuations and, most importantly, can be applied without any restrictions on the true valuations. At the core of our new auction design is a modular price update rule that only relies on samples from the auctioneer's valuation model, rather than a specific analytic form as used by \citet{brero2018Bayes}. We provide a new Bayesian interpretation of the price update problem as computing the most likely clearing prices given the current valuation model. This naturally leads to an Expectation-Maximization (EM) algorithm to compute modal prices, where valuations are latent variables. The key technical contributions to implement EM are (i) a generative process to sample from the joint distribution of prices and valuations (the expectation) and (ii) linear programming to optimize the approximate log likelihood (the maximization).
 	
 	We evaluate our general Bayesian CA on instances from the Combinatorial Auction Test Suite (CATS), a widely used instance generator for CAs~\citep{leyton2000towards}. We first consider single-minded valuations, and compare against the~\citet{brero2018Bayes} design. The performance of our general Bayesian CA design matches theirs in terms of clearing percentage and speed of convergence, even though their design is specialized to single-minded bidders.
 	Next, we evaluate our design in settings with general valuations, where we compare it against two very powerful benchmarks that use a subgradient CA design with a non-monotonic price update rule. Our results show that, on average (across multiple CATS domains), our general Bayesian CA outperforms the benchmark auctions in terms of clearing percentage and convergence speed.
 	%
 	
 	\paragraph{Practical Considerations and Incentives.} One can view our Bayesian iterative CA as a possible replacement for the clock phase of the CCA. Of course, in practice, many other questions (beyond the price update rule) are also important. For example, to induce (approximately) truthful bidding in the clock phase, the design of good \emph{activity rules} play a major role \citep{ausubel2017practical}. Furthermore, the exact payment rule used in the supplementary round is also important, and researchers have argued that the use of the Vickrey-nearest payment rule, while not strategyproof, induces good incentives in practice \citep{Cramton2013SpectrumAuctionDesign}. Our Bayesian CA, like the clock phase of the CCA, is not strategyproof. However, if our design were used in practice in a full combinatorial auction design, then we envision that one would also use activity rules, and suitably-designed payment rules, to induce good incentives. For this reason, we consider the incentive problem to be orthogonal to the price update problem. Thus, for the remainder of this paper, we follow prior work (e.g., \citet{parkes1999bundle}) and assume that bidders follow myopic best-response (truthful) bidding throughout the auction.
 	
 	%
 	%
 	%
 	

 	\section{Preliminaries}
 	\label{sec:preliminaries}	
 	The basic problem solved by an iterative combinatorial auction
 	is the allocation of a set of items, owned by a seller, among
 	a set of buyers who will place bids for the items during the auction. Let $m$ be the number of items and $n$ be the
 	number of buyers. The key features of the problem are that the
 	items are indivisible and that bidders have preferences over
 	sets of items, called bundles. We represent a bundle
 	using an $m$-dimensional indicator vector for the items it
 	contains and identify the set of bundles as $\X=\{0,1\}^m$.
 	We represent the preferences of each bidder $i$ with a
 	non-negative valuation function
 	$v_i: \mathcal{X} \rightarrow \Reals_{+}$ that is private
 	knowledge of the bidder. Thus, for each bundle
 	$x \in \mathcal{X}$, $v_i(x)$ represents the willigness to
 	pay, or \emph{value}, of bidder $i$ for obtaining $x$. We
 	denote a generic valuation profile in a setting with $n$
 	bidders as $v=(v_1, \ldots, v_n)$.
 	We assume that bidders have no value for the null bundle, i.e., $v_i(\emptyset) = 0$, and we assume free disposal, which implies that $v_i(x') \geq v_i(x)$ for all $x' \geq x$.
 	
 	At a high level, our goal is to design a combinatorial auction that computes an allocation that maximizes the total value to the
 	bidders. An \textit{iterative} combinatorial auction proceeds over
 	rounds, updating a provisional allocation of items to bidders
 	as new information about their valuations is obtained (via the
 	bidding), and updating prices over the items to guide the
 	bidding process. Accordingly, we next cover the key concepts
 	and definitions around allocations and prices.
 	
 	An \emph{allocation} is a vector of bundles,
 	$a = (a_1, \dots , a_n)$, with $a_i$ being the bundle that
 	bidder $i$ obtains. An allocation is \emph{feasible} if it respects
 	the supply constraints that each item goes to at most one
 	bidder.\footnote{We assume that there is one unit of each item for simplicity, but our work extends to multiple units without complications.} Let $\F \subset \X^n$
 	denote the set of feasible allocations. The total value of
 	an allocation $a$, given valuation profile $v$, is defined as
 	\begin{equation}
 	V(a; v) = \sum_{i\in[n]} v_i(a_i),
 	\end{equation}
 	where the notation $[n]$ refers to the index set $\{1,\ldots,n\}$.
 	An allocation $a^*$ is \emph{efficient} if
 	$a^* \in \argmax_{a \in \F} V(a; v)$.
 	In words, the allocation is efficient if it is feasible and maximizes the total value to the bidders.
 	
 	An iterative auction maintains $\emph{prices}$ over bundles of items, which are represented by a price function $\theta: \X \rightarrow \Reals_+$ assigning a price $\theta(x) $ to each bundle $x\in\X$. Even though our design can incorporate any kind of price function $\theta$, our implementations will only maintain prices over items, which are
 	represented by a non-negative vector $p \in \Reals^m_+$; this
 	induces a price function over bundles given by $\theta(x) = \sum_{j \in [m]} p_j x_j$.  Item prices are commonly used in practice as they are very intuitive and simple for the bidders to parse (see, e.g., \citet{ausubel2006clock}).\footnote{We emphasize that, although the framework generalizes conceptually to any kind of price function $\theta$, complex price structures may bring additional challenges from a computational standpoint.}
 	
 	Given bundle prices
 	$\theta$, the \emph{utility} of bundle $x$ to bidder $i$ is
 	$v_i(x) - \theta(x)$. The bidder's \emph{indirect utility} at
 	prices $\theta$ is
 	\begin{equation} \label{eq:indirect-bidder}
 	\textstyle{	U(\theta; v_i) = \max_{x\in \X} \{ v_i(x) - \theta(x)\}}
 	\end{equation}
 	i.e., the maximum utility that bidder $i$ can achieve by choosing among bundles from $\X$.
 	
 	On the seller side, the \emph{revenue} of an allocation $a$ at prices $\theta$ is $\sum_{i \in [n]} \theta(a_i)$. The
 	seller's \emph{indirect revenue} function is
 	\begin{equation} \label{eq:indirect-seller}
 	\textstyle{R(\theta) = \max_{a \in \F} \{ \sum_{i \in [n]} \theta(a_i)\}}
 	\end{equation}
 	i.e., the maximum revenue that the seller can achieve among all feasible allocations.
 	
 	
 	\paragraph{Market Clearing.}
 	
 	We are now in a position to define the central concept in this
 	paper.
 	\begin{definition}
 		Prices $\theta$ are \emph{clearing prices} if there exists a
 		feasible allocation $a \in \F$ such that, at bundle prices $\theta$, $a_i$ maximizes the utility of each bidder $i\in[n]$, and $a$ maximizes the seller's revenue over all feasible allocations.
 	\end{definition}
 	\noindent
 	We say that prices $\theta$
 	\emph{support} an allocation $a$ if the prices and allocation
 	satisfy the conditions given in the definition. The following
 	important facts about clearing prices follow from linear
 	programming duality (see~\citet{bikhchandani2002package} as a
 	standard reference for these results).
 	\begin{enumerate}
 		\item The allocation $a$ supported by clearing prices $\theta$
 		is efficient.
 		\item If prices $\theta$ support some allocation $a$, they
 		support every efficient allocation.
 		\item Clearing prices minimize the following objective function:
 		\begin{equation}\label{eq:unnormalizedClearingPot}
 		W(\theta; v) = \sum_{i \in [n]} U(\theta; v_i) + R(\theta).
 		\end{equation}	
 	\end{enumerate}
 	The first fact clarifies our interest in clearing prices: they
 	provide a certificate for efficiency. An iterative auction can
 	update prices and check the clearing condition by querying the bidders, thereby solving the allocation problem
 	without necessarily eliciting the bidders' complete
 	preferences. The second fact implies that it is possible to
 	speak of clearing prices without specific reference to the
 	allocation they support. The interpretation of clearing prices
 	as minimizers of~(\ref{eq:unnormalizedClearingPot}) in the third
 	fact will be central to our Bayesian approach.\footnote{We emphasize that none of the results above assert the \emph{existence}
 		of clearing prices of the form $\theta(x) = \sum_{j \in [m]} p_j x_j$. The objective
 		in~(\ref{eq:unnormalizedClearingPot}) is implicitly defined in terms of item prices $p$, and for general valuations there
 		may be no item prices that satisfy the clearing price
 		condition~\citep{gul2000english}.}

 	%
 	
 	\paragraph{Clearing Potential.} By linear programming duality~\citep{bikhchandani2001linear}, we have
 	$$ W(\theta; v) \geq V(a; v) $$
 	for all prices $\theta$ and feasible allocations $a$, and the
 	inequality is tight if and only if $\theta$ are clearing and
 	$a$ is efficient. In the following, we will therefore make use
 	of a ``normalized'' version of~(\ref{eq:unnormalizedClearingPot}):
 	\begin{equation}\label{eq:clearingPot}
 	\hat W(\theta;  v) = \bigg( W(\theta;  v) - V(a^*; v) \bigg)\ge  0.
 	\end{equation}	
 	It is useful to view~\eqref{eq:clearingPot} as a potential function
 	that quantifies how close prices $\theta$ are to clearing prices; the potential can reach 0 only if there exist clearing prices.
 	%
 	%
 	We refer to function~\eqref{eq:clearingPot} as
 	the \textit{clearing potential} for the valuation profile $ v
 	$ which will capture, in a formal sense, how likely a price
 	function $\theta$ is to clearing the valuation profile
 	$v$ within our Bayesian framework.
 	%
 	
 	
 	\section{The Bayesian Auction Framework}
 	\label{sec:BayesianMechanism}
 	
 	\begin{algorithm}[t!]
 		\SetEndCharOfAlgoLine{.}		
 		\textbf{Input: }{Prior beliefs $Q^0(v)$.}\\
 		$\theta^0=$ initial prices, $t = 0$\;
 		\Repeat{$a_i^t = b_i^t$ \emph{for each bidder} $i\in [n]$}{
 			$t \leftarrow t+1$\;
 			Observe each bidder $i$'s demanded bundle $b_i^t$ at $\theta^{t-1}$\;
 			Compute revenue-maximizing allocation $a^t$ at $\theta^{t-1}$\;
 			\textbf{Belief update:} Use each bidder $i$'s demand $b_i^t$ and $Q^{t-1}(v)$ to derive $Q^{t}(v)$\;
 			\textbf{Price update:} Use $Q^{t}(v)$ to derive new prices $\theta^{t}$\;
 		}	
 		\caption{Bayesian Auction Framework}
 		\label{alg:Bayes_Elicitation}
 	\end{algorithm}	
 	
 	We now describe the Bayesian auction framework introduced by \citet{brero2018Bayes} (see  Algorithm~\ref{alg:Bayes_Elicitation}). At the beginning of the auction, the auctioneer has a prior belief over bidder valuations which is modeled via the probability density function $Q^0(v)$. First, some initial prices $\theta^0$ are quoted (Line 2). It is typical to let $\theta^0$ be ``null prices'' which assign price zero to each bundle.
 	At each round $t$, the auctioneer observes the demand $b_i^t$ of each bidder $i$ at prices $\theta^{t-1}$ (Line 5), and computes a revenue maximizing allocation $a^t$ at prices $\theta^{t-1}$ (Line 6). The demand observations are used to update the beliefs $Q^{t-1}(v)$ to $Q^{t}(v)$ (Line 7), and new prices $\theta^{t}$ reflecting new beliefs are quoted (Line 8). This procedure is iterated until the revenue maximizing allocation matches bidders' demand (Line 9), which indicates that the elicitation has determined clearing prices.
 	
 	In this paper, we build on the framework introduced by \citet{brero2018Bayes} and generalize their approach to (i) handle any kind of priors and (ii) apply it to settings with no restrictions on the bidders' valuations. This requires completely new instantiations of the \emph{belief update rule} and the \emph{price update rule}, which are the main contributions of our paper, and which we describe in the following two sections.
 	
 	
 	\section{Belief Update Rule}\label{sec:beliefUpdate}
 	In this section, we describe our belief modeling and updating rule, based on Gaussian approximations of the belief distributions (which proved effective in our experiments). We emphasize, however, that the belief update component of the framework is modular and could accommodate other methods like expectation-propagation or non-Gaussian models. A complete description of the rule is provided in Algorithm~\ref{alg:Belief_Update}.
 	
 	The auctioneer first models the belief distribution via the probability density function  $Q(v)=\prod_{i=1}^n Q_i(v_i)$. As the rounds progress, each bidder bids on a finite number of bundles (at most one new bundle per round). Let $B_i$ be the set of bundles that bidder $i$ has bid on up to the current round. We model $Q_i$ as
 	\begin{equation}
 	Q_i(v_i) = \prod_{b_i \in  B_i} Q_i(v_i(b_i)).
 	\end{equation}
 	Note that $Q_i(\cdot)$ assigns equal probability to any two valuations that assign the same values to bundles in $B_i$. {We model the density over each $v_i(b_i)$ using a Gaussian}:
 	\begin{equation}
 	Q_i(v_i(b_i))=\mathcal N(\mu_i(b_i), \sigma_i(b_i)),
 	\end{equation}
 	where $\mathcal N(\mu,\sigma)$ denotes the density function of the Gaussian distribution with mean $\mu$ and standard deviation $\sigma$. By keeping track of separate, independent values for different bundles bid on, the auctioneer is effectively modeling each bidder's preferences using a multi-minded valuation. However, as this is just a model, this does not imply that the bidders' true valuations are multi-minded over a fixed set of bundles.
 	
 	\begin{algorithm}[t!]
 		\SetEndCharOfAlgoLine{.}		
 		\textbf{Input: }{Beliefs $Q^{t-1}(v)$, demand $b_i^t$ of each bidder $i\in[n]$\;}
 		\ForEach{$i\in [n]$}{
 			\eIf{$b_i^t \neq \emptyset$}
 			{
 				\If{$b_i^t\notin B_i$}{
 					$B_i\leftarrow B_i\cup \{b_i^t\}$\;
 					$Q_i^{t-1}(v_i(b_i^t)) = Q({v}^0(b_i^t)).$
 				}
 				$\begin{array}{l} Q_i^{t}(v_i(b_i^t)) \approx \\
 				\hspace{4mm} \Phi(\beta(v_i(b_i^t)-\theta^{t-1}(b_i^t)))\cdot Q_i^{t-1}(v_i(b_i^t)).\end{array}$ \\
 				\ForEach{$b_i\in B_i\setminus{\{b_i^t\}}$}{
 					$Q_i^{t}(v_i(b_i)) = Q_i^{t-1}(v_i(b_i))$\;
 				}
 			}
 			{
 				\ForEach{$b_i\in B_i$}{
 					$\begin{array}{l}Q_i^t(v_i(b_i)) \approx \\
 					\hspace{4mm} \Phi(\beta(\theta^{t-1}(b_i)-v_i(b_i)))\cdot Q_i^{t-1}(v_i(b_i)).\end{array}$
 				}
 			}
 		}
 		\textbf{Output: }{Updated beliefs $Q^{t}(v)$.}
 		\caption{Belief Update Rule}
 		\label{alg:Belief_Update}
 	\end{algorithm}	
 	
 	We now describe how the auctioneer updates $Q^{t-1}(v)$ to $Q^t(v)$ given the bids observed at round $t$. We assume the auctioneer maintains a Gaussian distribution with density $Q(v^0(x)) = \mathcal{N}(\mu^0(x), \sigma^0(x))$ over the value a generic bidder may have for every bundle $x \in \X$.
 	To update beliefs about bidder valuations given their bids, the auctioneer needs a probabilistic model of buyer bidding. According to myopic best-response bidding, at each round $t$, bidder $i$ would report a utility-maximizing bundle $b_i^t \in \X$ at current prices $\theta^{t-1}$. 
 	{Formally, we have that, for each $x \in \mathcal X$,
 		\begin{equation}\label{eq:properBeliefUpdate}
	 		v_i(b^t_i) - \theta^{t-1}(b^t_i) \ge v_i(x) - \theta^{t-1}(x).
 		\end{equation}
 		To properly capture this, the auctioneer needs to update her beliefs over each $v_i(x)$ together with $v_i(b^t_i)$, which is a relatively complex operation for each bidder at each round. In this work we introduce a first, simple approach where the information provided by each bid $b_i^t$ is only used to update the auctioneer's beliefs over value $v_i(b^t_i)$ using that $v_i(b^t_i) - \theta^{t-1}(b^t_i) \ge 0$ (which is the special case of Equation~\eqref{eq:properBeliefUpdate} when $x$ is the empty bundle). According to the myopic best-response model, buyer $i$ would bid on $b_i^t$ with probability 1 if $v_i(b^t_i)>\theta^{t-1}(b^t_i)$, and 0 if $v_i(b^t_i)<\theta^{t-1}(b^t_i)$ (there are no specifications of bidder $i$'s behavior when $v_i(b^t_i)=\theta^{t-1}(b^t_i)$).} 
 	Note that this kind of bidding model is incompatible with Gaussian modeling because it contradicts full support: all bundle values $v_i(b_i^t) < \theta^{t-1}(b_i^t)$ must have probability 0 in the posterior. To account for this, we relax this sharp myopic best-response model to probit best-response, a common random utility model under which the probability of bidding on a bundle is proportional to its utility~\citep{train2009discrete}. 
 	Specifically, we set the probability that bidder $i$ bids on $b_i^t$ at prices $\theta^{t-1}$ proportional to
 	\begin{equation}
 	\Phi(\beta(v_i(b^t_i)-\theta^{t-1}(b^t_i))),
 	\end{equation}
 	where $\Phi$ is the cumulative distribution function of the standard Gaussian distribution {and 
 	$\beta>0$ is a scalar parameter that controls the extent to which $\Phi$ approximates myopic best-response bidding.} 
%
 	
 	Given a bid on bundle $b_i^t \neq \emptyset$ in round $t$, the auctioneer first records the bundle in $B_i$ if not already present (Line~5) and sets $Q^{t-1}_i(v_i(b_i^t))$ to $Q(v^0(b_i^t))$ (Line 6).	
 	The belief distribution over value $v_i(b_i^t)$ is updated to
 	\begin{eqnarray*}
 		Q_i^{t}(v_i(b_i^t)) & = & \mathcal{N}(\mu_i^{t}(b_i^t), \sigma_i^{t}(b_i^t)) \\
 		& \hspace{-24mm} \approx & \hspace{-12mm} \Phi(\beta(v_i(b^t_i)-\theta^{t-1}(b^t_i)))\cdot Q_i^{t-1}(v_i(b_i^t)) \\
 		& \hspace{-24mm} = & \hspace{-12mm} \Phi(\beta(v_i(b^t_i)-\theta^{t-1}(b^t_i)))\cdot \mathcal N(\mu_i^{t-1}(b^t_i), \sigma_i^{t-1}(b^t_i))
 	\end{eqnarray*}
 	(Line 8). To approximate the right-most term with a Gaussian, we use simple moment matching, setting the mean and variance of the updated Gaussian to the mean and variance of the right-hand term, which can be analytically evaluated for the product of the Gaussian cumulative distribution function $\Phi$ and probability density function $\mathcal N$ (see for instance~\citet{williams2006gaussian}).
 	This is a common online Bayesian updating scheme known as \textit{assumed density filtering}, a special case of expectation-propagation~\citep{opper1998bayesian,minka2001family}.
 	If $b_i^t = \emptyset$, then we update the value of \emph{every} bundle $b_i \in B_i$ with an analogous formula (Line 14), except that the probit term is replaced with
 	\begin{equation}
 	\Phi(\beta(\theta^{t-1}(b_i)-v_i(b_i)))	
 	\end{equation}
 	to reflect the fact that declining to bid indicates that $v_i(b_i) \leq \theta^{t-1}(b_i)$ for each bundle $b_i \in B_i$.
 	
 	\section{Price Update Rule}\label{sec:InferringClearingPrices}
 	
 	In this section, we describe our price update rule. A complete description of the rule is provided in Algorithm~\ref{alg:Price_Update}. The key challenge we address in this section is how to derive ask prices from beliefs over bidders' valuations. We transform this problem into finding the mode of a suitably-defined probability distribution over prices, and we then develop a practical approach to computing the mode via Monte Carlo Expectation Maximization.
 	
 	We seek a probability density function $P(\theta)$ over prices whose maxima are equal to those prices that will most likely be clearing under $Q(v)$. As a first attempt, consider an induced density function over clearing prices as given by
 	\begin{equation}\label{eq:pdfPrices}
 	P(\theta)  \propto \int_v \ones\{\hat{W}(\theta;  v) = 0\}\,Q(v)\,dv,
 	\end{equation}		
 	with $P(\theta, v) \propto \ones\{\hat{W}(\theta;  v) = 0\}\,Q(v)$ as the associated joint density function.
 	Recall that $\hat{W}(\theta;  v) = 0$ if and only if prices $\theta$ are clearing for valuations $v$. Thus, under function~\eqref{eq:pdfPrices}, prices $\theta$ get assigned all the probability density of the configurations $v$ for which they represent clearing prices.
 	
 	Although this approach is natural from a conceptual standpoint, it may lead to problems when the Bayesian auction uses specific price structures (e.g., item prices) that cannot clear any valuation in the support of $Q(v)$. It is then useful to introduce price distributions such that, for each price function $\theta$, $P(\theta, v)>0$ for all possible configurations $v$.
 	To obtain a suitable price density function we approximate 	
 	$P(\theta)$ with 	
 	\begin{equation}\label{eq:approxPdfPrices}
 	P_\lambda (\theta)  \propto \int_v e^{-\lambda \hat W(\theta;  v)}\,Q(v)\,dv.
 	\end{equation}
 	The approximation is motivated by the following proposition.
 	
 	\begin{proposition}\label{prop:approxPdfPrices}
 		Assume that $Q(v)$ allows us to define a probability density function over prices via Equation~\eqref{eq:pdfPrices}. Then, for every price function $\theta$,
 		\begin{equation}
 		\lim_{\lambda \to \infty} P_\lambda (\theta) = P (\theta).
 		\end{equation}
 		\begin{proof}
 			We prove the convergence of $P_\lambda (\theta)$ to $P (\theta)$ by separately showing the convergence of the numerator
 			$\int_v e^{-\lambda \hat W(\theta;  v)}\,Q(v) \,dv$ to  $\int_v  \ones\{\hat{W}(\theta;  v) = 0\}\,Q(v)\,dv$
 			and of the normalizing constant $\int_\theta \int_v e^{-\lambda \hat W(\theta;  v)}\,Q(v) \,dv\, d\theta$ to $\int_\theta \int_v  \ones\{\hat{W}(\theta;  v) = 0\}\,Q(v)\,dv\, d\theta$.  We will only show the convergence of the numerator, as the convergence of the normalizing constant follows from very similar reasoning.
 			
 			Given that $\hat W(\theta;  v)\ge0$ for any $\theta$ and $\hat W(\theta;  v)=0$ only if $\theta$ is clearing for $v$, we have that, for each $v$, $\theta$,
 			\begin{equation}
 			\lim_{\lambda \to \infty} e^{-\lambda \hat W(\theta;  v)}\,Q(v) =  \ones\{\hat{W}(\theta;  v) = 0\}\,Q(v).
 			\end{equation}				
 			To achieve convergence in the integral form, we note that, as $e^{-\lambda \hat W(\theta;  v)}$ varies between 0 and 1, $e^{-\lambda \hat W(\theta;  v)}\,Q(v)$ is bounded by the integrable probability density function $Q(v)$. This allows us to obtain convergence in the integral form via Lebesgue's Dominated Convergence Theorem.
 		\end{proof}					
 	\end{proposition}
 	
 	\noindent	
 	We can now interpret $P_{\lambda}(\theta)$ as the marginal probability density function of
 	\begin{equation} \label{eq:join-lambda-density}
 	P_{\lambda}(v,\theta) \propto e^{-\lambda \hat W(\theta; v)}\,Q(v).
 	\end{equation}
 	The standard technique for optimizing density functions like~(\ref{eq:approxPdfPrices}), where latent variables are marginalized out, is Expectation Maximization (EM). The EM algorithm applied to $P_{\lambda}$ takes the following form:
 	\begin{itemize}
 		\item  \underline{E step}: At each step $\tau$, we compute
 		\begin{equation*}
 		\mathbb E_{P_{\lambda}( v|\theta^{\tau-1})} \bigg[\log P_{\lambda}( v,\theta) \bigg ],
 		\end{equation*}
 		where
 		\begin{equation}\label{val-lambda-distrib}
 		P_{\lambda}( v\cond \theta ) = \frac{\ds e^{-\lambda \hat W(\theta;  v)}\,Q(v)}{\ds \int_{ v'}e^{-\lambda \hat W(\theta;  v')}\,Q( v')\,d  v'}.
 		\end{equation}		
 		\item \underline{M step}: Compute new prices
 		\begin{equation*}
 		\theta^{\tau} \in \argmax_\theta \mathbb E_{P_{\lambda}( v|\theta^{\tau-1})} \bigg[\log P_{\lambda}( v,\theta) \bigg ].
 		\end{equation*}
 	\end{itemize}
 	In general, it may be infeasible to derive closed formulas for the expectation defined in the E step. We then use the Monte Carlo version of the EM algorithm introduced by \citet{wei1990monte}. For the E step, we provide a sampling method that correctly samples from~$P_{\lambda}(v | \theta)$, the conditional probability density function of valuations obtained from~(\ref{eq:join-lambda-density}). 	
 	For the M step, we use linear programming to optimize the objective, given valuation samples.
 	
 	\begin{algorithm}[t]
 		\SetEndCharOfAlgoLine{.}		
 		\textbf{Input: }{Current beliefs $Q(v)$\;}
 		$\theta^0=$ initial prices, $\tau = 0$\;
 		\Repeat{$||\theta^{\tau} - \theta^{\tau-1} ||/|| \theta^{\tau-1} || < \varepsilon$ }{
 			$\tau \leftarrow \tau+1$\; 		
 			\ForEach{$k\in [\ell]$}{
 				\Repeat{$ \emph{\texttt{resample}} = 0$}{
 					Set \texttt{resample} = 0\;				
 					Draw $v^{(k)}$ from $Q(v)$\;
 					Set \texttt{resample} = 1 with probability $1-e^{-\lambda \hat W(\theta^{\tau-1};  v^{(k)})}$\;
 				}
 			}
 			Compute $\theta^{\tau}\in \argmin_\theta \sum_{k\in [\ell]} \hat W(\theta;  v^{(k)})$\;
 		}
 		\textbf{Output: }{Prices $\theta^{\tau}$.}
 		\caption{Price Update Rule}
 		\label{alg:Price_Update}
 	\end{algorithm}
 	
 	\paragraph{Monte Carlo EM.} Our Monte Carlo EM algorithm works as follows: At each step $\tau$,
 	\begin{itemize}
 		\item Draw samples $v^{(1)},..., v^{(\ell)}$ from $P_\lambda(v \cond \theta^{\tau-1})$ (Lines~{5-11}).
 		\item Compute new prices
 		\begin{eqnarray*}
 			\theta^{\tau}
 			& \in & \argmax_\theta \frac{1}{\ell} \sum_{k\in[\ell]} \log P_{\lambda}(v^{(k)},\theta) \\ \\
 			& \in & \argmin_{\theta}\sum_{k\in[\ell]} \hat W(\theta;  v^{(k)}) \\ \\
 			& \in & \argmin_{\theta}\sum_{k\in[\ell]} \sum_{i\in [n]} U(\theta;v_i^{(k)}) + \ell R(\theta)
 		\end{eqnarray*}	
 	\end{itemize}
 	where the last step follows from \eqref{eq:unnormalizedClearingPot} and \eqref{eq:clearingPot} (Line 12). Each $\theta^{\tau}$ can be derived via the following linear program (LP): 
 	
 	\begin{equation}\label{eq:LP}
 	\hspace{-32mm}\underset{\theta, \pi_{ik}\ge 0, \pi\ge 0}{\text{minimize}} \quad \sum_{k\in[\ell]} \sum_{i\in[n]} \pi_{ik} + \ell \,  \pi
 	\end{equation}
 	\begin{equation*}
 	\text{s.t.}\quad \pi_{ik} \ge v_i^{(k)}(b_{i})-\theta(b_{i})\quad  \forall i\in [n], k\in [\ell], b_{i}\in B_i,
 	\end{equation*}
 	\vspace{1mm}
 	\begin{equation*}
 	\hspace{-27mm} \pi \ge \sum_{i\in[n]} \theta(a_i) \quad \forall a \in \mathcal F.
 	\end{equation*}
 	
 	Note that, at any optimal solution for LP~\eqref{eq:LP}, each variable $\pi_{ik}$ equals the indirect utility~\eqref{eq:indirect-bidder} of bidder $i$ in sample $k$ at prices $\theta$, while $\pi$ equals the seller's indirect revenue~\eqref{eq:indirect-seller}.
 	Under item prices, $\theta$ can be parametrized via $m$ variables $p_j\ge 0$, and the last set of constraints reduces to $ \pi \ge \sum_{j\in[m]} p_j$. Furthermore, as discussed in Section~\ref{sec:beliefUpdate}, the size of each $B_i$ cannot be larger than the number of rounds. However, note that both the number of constraints and the number of variables $\pi_{ik}$ are proportional to the number of samples $\ell$. In our experimental evaluation we confirm that it is possible to choose the sample size $\ell$ to achieve both a good expectation approximation and a tractable LP size.
 	
 	\paragraph{Sampling Valuations from Posterior (Lines 5-11).}
 	To sample from the posterior $P_{\lambda}(v \cond \theta)$, we use the following generative model of bidder valuations. This generative model is an extension of the one provided by~\citet{brero2018Bayes}, itself inspired by~\citet{sollich2002bayesian}.	
 	
 	\begin{definition}\label{def:generativeModel}[Generative Model] The generative model over bidder valuations, given prices $\theta$, is defined by the following procedure:
 		\begin{itemize}
 			\item Draw $v^{(k)}$ from $Q(v)$ (Line~8).
 			\item Resample with probability $1- e^{-\lambda \hat W(\theta;  v^{(k)})}$ (Line~9). \footnote{Note that, to determine whether to resample $v^{(k)}$, one needs to compute the optimal social welfare $V(a^*; v^{(k)})$ via $\hat W(\theta;  v^{(k)})$ (see Equation~\eqref{eq:clearingPot}), which may require a computationally costly optimization. This is not the case in our implementation as each $v^{(k)}_i $ is a multi-minded valuation over a small set of bundles $B_i$. Alternatively, one can use the ``unnormalized'' $W(\theta;  v^{(k)})$ as a proxy for $\hat W(\theta;  v^{(k)})$, as in the generative model proposed by \citet{brero2018Bayes} for single-minded bidders.
 				When the optimal social welfare is not varying too much across different samples,	
 				this trick provides a good approximation of our generative model. Furthermore, it also did not prevent \citet{brero2018Bayes} from obtaining very competitive results.}
 		\end{itemize}		
 	\end{definition}
 	
 	The following proposition confirms that the generative model correctly generates samples from~(\ref{val-lambda-distrib}).
 	\begin{proposition}\label{prop:sampling}
 		The samples generated by our generative model have probability defined via the density function
 		\begin{equation}
 		P_{gen}(v)  \propto e^{-\lambda \hat W(\theta;  v)}\,Q(v),
 		\end{equation}	
 		which corresponds to $P_{\lambda}( v \cond \theta )$.
 	\end{proposition}	
 	\begin{proof}
 		We denote the probability of resampling as
 		\begin{equation}
 		r=1-\int_v e^{-\lambda \hat W(\theta;  v)}\,Q(v)\,dv.
 		\end{equation}
 		The probability that $v$ will be drawn after $h$ attempts is
 		\begin{equation}
 		P(v, h) =  e^{-\lambda \hat W(\theta;  v)}\,Q(v)\, r^h.
 		\end{equation}
 		Thus, we have that 		
 		\begin{equation}
 		P_{gen}(v) =\sum_{h=0}^\infty P(v, h) \propto  e^{-\lambda \hat W(\theta;  v)}\,Q(v).
 		\end{equation}			
 	\end{proof}					
 	\noindent
 	The $\lambda$ relaxation of the price density function has interesting computational implications in our sampling procedure. The larger the $\lambda$ (i.e., the better the approximation of Equation~\eqref{eq:pdfPrices}), the larger the probability of resampling. Thus, using smaller $\lambda$ will speed up the sampling process at a cost of lower accuracy. From this perspective, $\lambda$ can serve as an annealing parameter that should be increased as the optimal solution is approached.	
 	However, while a larger $\lambda$ increases the probability of finding clearing prices under density function $Q(v)$, it does not necessarily lead to better clearing performances in our auctions. Indeed, $Q(v)$ is affected by the auctioneer's prior beliefs which may not be accurate. In particular, when the observed bids are unlikely under $Q(v)$, it can be useful to decrease $\lambda$ from round to round. In our experimental evaluations, we will simply use $\lambda=1$ and scale valuations between 0 and 10 as (implicitly) done by \citet{brero2018Bayes}. This also keeps our computation practical. We defer a detailed analysis of this parameter to future work.
 	
 	\begin{figure*}[t]
 		\centering
 		\includegraphics[width=1\linewidth]{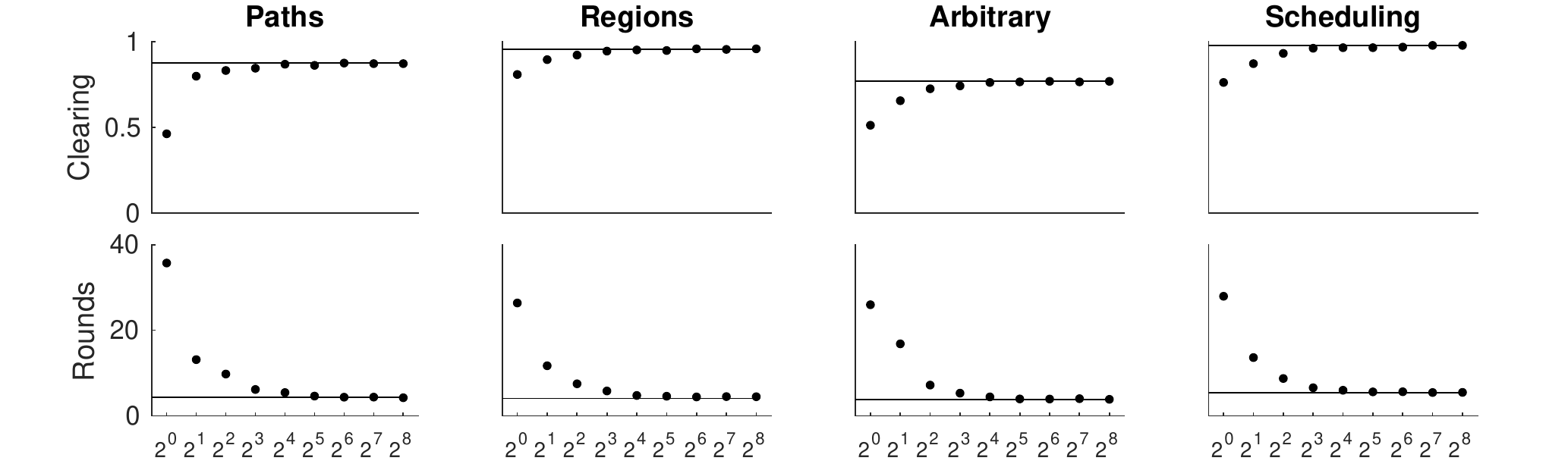}
 		\caption{Average number of cleared instances and rounds under different number of samples $\ell$ in our Monte Carlo Expectation Maximization Algorithm. The horizontal lines indicate the baseline performance of Bayes$^\text{SM}$ \citep{brero2018Bayes}.}
 		\label{fig:samplesMCEM}
 	\end{figure*}
 	
 	\begin{table*}[b]
 		\setlength\tabcolsep{2pt}
 		\begin{center}
 			\begin{tabular}{|c||c|c|c|c|c|c|c|c|}
 				\hline	&
 				\multicolumn{2}{|c|}{\textbf{Paths}} &
 				\multicolumn{2}{|c|}{\textbf{Regions}} &
 				\multicolumn{2}{|c|}{\textbf{Arbitrary}} & 					
 				\multicolumn{2}{|c|}{\textbf{Scheduling}}\\
 				\cline{2-9}
 				& Clearing & Rounds & Clearing & Rounds& Clearing & Rounds& Clearing & Rounds\\
 				\hline
 				\hline
 				SG-Auction$^\text{D}$
 				& 84\% & 19.5 (1.0)
 				& 89\% &  24.8 (1.2)
 				& 65\% &  35.1 (1.8)
 				& 94\% & 21.0 (1.2)
 				\\
 				\hline
 				SG-Auction$^\text{I}$ 	
 				& 88\% &  8.6 (0.4)
 				& 95\% & 15.9 (0.7)
 				& 75\% &  21.3 (1.0)
 				& 97\% & 11.9 (0.6)
 				\\
 				\hline
 				Bayes$^\text{GEN}$
 				& 88\% & 5.2 (0.4)
 				& 96\% & 4.6 (0.3)
 				& 77\% & 4.5 (0.3)
 				& 98\% & 6.3 (0.3)
 				\\
 				\hline
 				Bayes$^\text{SM}$ \citep{brero2018Bayes}
 				& 88\% & 4.9 (0.2)
 				& 96\% & 4.3 (0.2)
 				& 77\% & 4.2 (0.3)
 				& 98\% & 6.1 (0.3)
 				\\
 				\hline
 			\end{tabular}
 			\caption{Comparison of different auction designs in settings with single-minded bidders. Clearing results are averaged over 300 auction instances. Rounds results are averaged over those instances that were cleared by all four auctions (which always included more than 100 instances for each distribution). Standard errors are reported in parentheses.}
 			\label{tab:singleMinded}
 		\end{center}
 	\end{table*}	
 	
 	\section{Empirical Evaluation}
 	\label{sec:empirical-evaluation}
 	
 	We evaluate our Bayesian auction via two kinds of experiments. In the first set of experiments we consider single-minded settings where we compare our auction design against the one proposed by \citet{brero2018Bayes}. These experiments are meant to determine how many samples at each step of our Monte Carlo algorithm we need to draw to match their results. In the second set of experiments we consider multi-minded settings. Here, we will compare our auction design against non-Bayesian baselines.
 	
 	\subsection{Experiment Set-up}
 	\paragraph{Settings.} 
 	We evaluate our Bayesian auction on instances with 12 items and 10 bidders. These instances are sampled from four distributions provided by the Combinatorial Auction Test Suite (CATS): \emph{paths}, \emph{regions}, \emph{arbitrary}, and \emph{scheduling}~\citep{leyton2000towards}.
 	Each instance is generated as follows. First, we generate an input file with 1000 bids over 12 items. Then we use these bids to generate a set of bidders.	
 	To generate multi-minded bidders, CATS assigns a dummy item to bids: bids sharing the same dummy item belong to the same bidder.
 	To generate single-minded bidders we simply ignore the dummy items. In each CATS file,
 	we partition the bidders into a \textit{training set} and a \textit{test set}.
 	The training set is used to generate the prior Gaussian density functions $Q(v^0(x)) = \mathcal{N}(\mu^0(x), \sigma^0(x))$ over bundle values $v^0(x)$ which are used to initialize the auctioneer beliefs in our auction. Specifically, we fit a linear regression model using a Gaussian process with a linear covariance function which predicts the value for a bundle as the sum of the predicted value of its items. The fit is performed using the publicly available GPML Matlab code \citep{williams2006gaussian}. Each bid of each bidder in the training set is considered as an observation. We generate the actual auction instance by sampling 10 bidders uniformly at random from the test set.	We repeat this process 300 times for each distribution to create 300 auction instances which we use for the evaluation of our auction designs.
 	
 	\paragraph{Upper Limit on Rounds.}
 	As mentioned in Section~\ref{sec:preliminaries}, the implementation of our Bayesian auction is based on item prices given by $\theta(x)=\sum_{j\in[m]} p_jx_j$, where $p_j \ge 0$. Item prices may not be expressive enough to support an efficient allocation in CATS instances \citep{gul2000english}. We therefore set a limit of 100 rounds for each elicitation run and record reaching this limit as a failure to clear the market.
 	Note that, under this limit, some instances will only be cleared by some of the auctions that we test. To avoid biases, we always compare the number of rounds on the instances cleared by all auctions we consider.\footnote{Alternatively, one could allow for more than 100 rounds on each instance where item clearing prices are found by any of the tested auctions. However, relaxing the cap on the number of rounds can lead to outliers with a very high number of rounds which can drastically affect our results.}
 	
 	
 	\begin{table*}[t!]
 		\setlength\tabcolsep{2pt}
 		\begin{center}
 			\begin{tabular}{|c||c|c|c|c|c|c|c|c|}
 				\hline &
 				\multicolumn{2}{|c|}{\textbf{Paths}} &
 				\multicolumn{2}{|c|}{\textbf{Regions}} &
 				\multicolumn{2}{|c|}{\textbf{Arbitrary}} & 					
 				\multicolumn{2}{|c|}{\textbf{Scheduling}}\\
 				\cline{2-9}
 				& Clearing & Rounds & Clearing & Rounds& Clearing & Rounds& Clearing & Rounds\\
 				\hline \hline
 				SG-Auction$^\text{D}$
 				& 46\% & 22.0 (1.44)	
 				& 75\% & 28.8 (1.4)
 				& 34\% & 36.2 (2.6)
 				& 51\% & 31.2 (2.2)
 				\\
 				\hline
 				SG-Auction$^\text{I}$ 	
 				& 49\% & 9.3 (0.6)
 				& 81\% & 19.0 (0.9)
 				& 42\% & 27.4 (2.1)
 				& 59\% & 21.2 (1.4)
 				\\
 				\hline
 				Bayes$^\text{GEN}$
 				& 47\% & 11.5 (1.3)
 				& 83\% & 8.3 (0.6)
 				& 47\% & 9.7 (0.6)
 				& 57\% & 18.8 (1.4)
 				\\
 				\hline
 			\end{tabular}
 			\caption{Comparison of different auction designs in settings with multi-minded bidders. Clearing results are averaged over 300 auction instances. Rounds results are averaged over instances that were cleared by all three auctions (which always included more than 100 instances for each distribution). Standard errors are reported in parentheses.}
 			\label{tab:clearingMultiMinded}
 		\end{center}
 	\end{table*}	
 	
 	\paragraph{Non-Bayesian Baselines.}
 	We compare our auction design against non-Bayesian baselines that are essentially clock auctions that increment prices according to excess demand, closely related to clock auctions used in practice like the CCA, except that prices are not forced to be monotone.
 	Because the Bayesian auctions are non-monotone, we consider non-monotone clock auctions a fairer (and stronger) comparison. The baseline clock auctions are parametrized by a single scalar positive step-size which determines the intensity of the price updates. We refer to these auctions as \textit{subgradient auctions} (SG-Auctions), as they can be viewed as subgradient descent methods for computing clearing prices.
 	
 	To optimize these subgradient auctions we run them 100 times on each setting, each time using a different step-size parameter spanning the interval from zero to the maximum bidder value. When then consider the following baselines:
 	\begin{itemize}
 		\item Subgradient Auction Distribution (SG-Auction$^D$): the subgradient auction using the step-size parameter that leads to the best performance \emph{on average} over the auction instances generated from any given distribution.
 		\item Subgradient Auction Instance (SG-Auction$^I$): the subgradient auction using the step-size parameter that leads to the best performance \emph{on each individual auction instance}.
 	\end{itemize}
 	Note that these baselines are designed to be extremely competitive compared to our Bayesian auctions. In particular, in SG-Auction$^I$, the auctioneer is allowed to run 100 different subgradient auctions on each instance and choose the one that cleared with the lowest number of rounds.
 	
 	
 	\begin{figure}[b!]
 		\begin{center}
 			\includegraphics[width=0.4\textwidth]{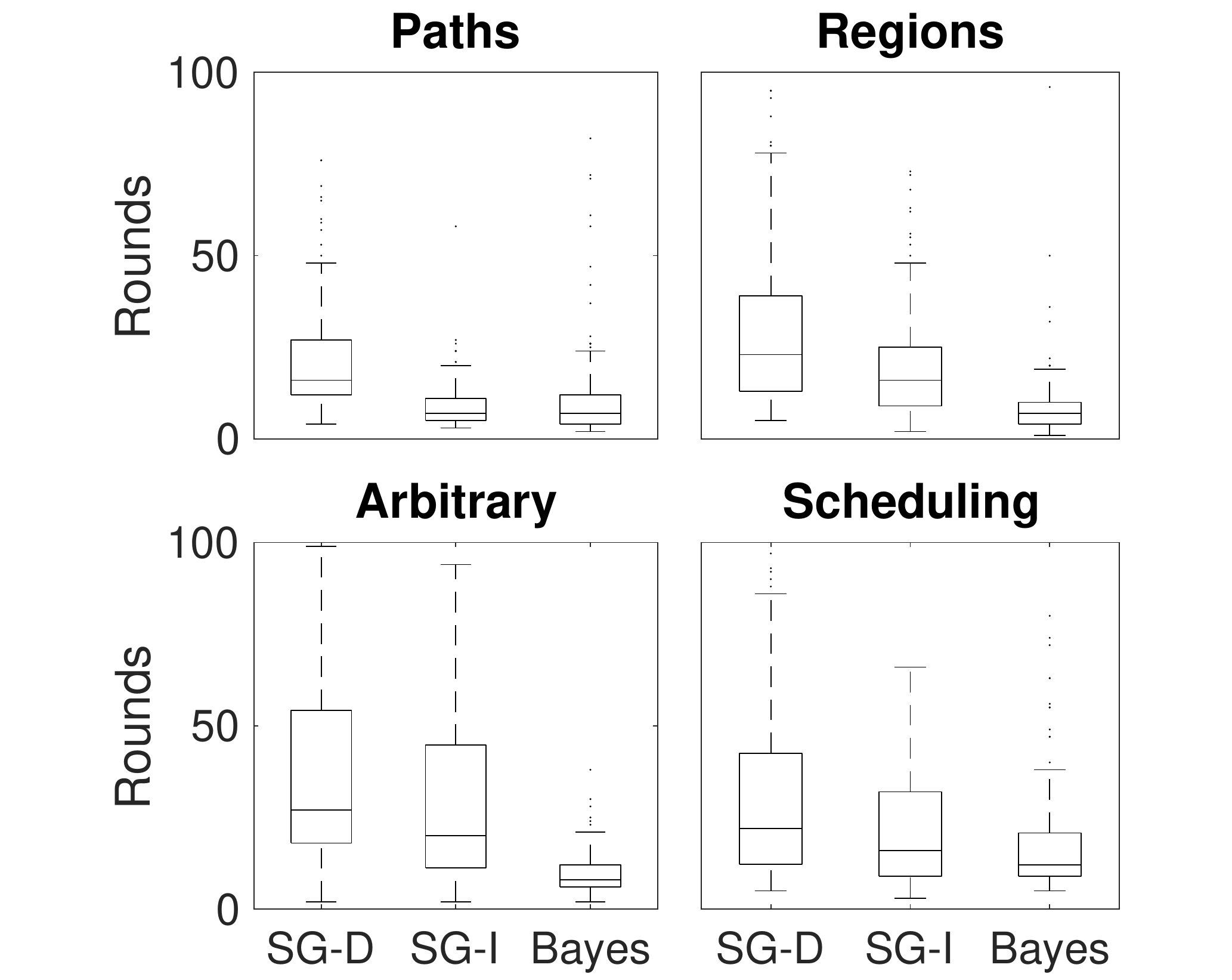}
 		\end{center}
 		\vspace{-0.25cm}
 		\caption{
 			Auction rounds under multi-minded bidders. The
 			box plot provides the first, second and third quartiles; the
 			whiskers are at the 10th and 90th percentile.}
 		\label{figure:CfrAuctionRounds}
 	\end{figure}

 	\subsection{Results for Single-Minded Settings}
 	We now compare our Bayesian auction, denoted Bayes$^\text{GEN}$, against the one proposed by \citet{brero2018Bayes} (which is limited to single-minded settings), denoted Bayes$^\text{SM}$, and against the non-Bayesian baselines.
 	
 	We first consider different versions of our auction design where we vary the number of samples $\ell$ used at each step of the Monte Carlo EM algorithm. As shown in Figure~\ref{fig:samplesMCEM}, our general Bayesian auction is competitive with Bayes$^\text{SM}$ \citep{brero2018Bayes} starting from $\ell = 2^4$ (note that this is true for all distributions even though they model very different domains). This low number of samples allows us to solve the linear program presented in \eqref{eq:LP} in a few milliseconds. For the remainder of this paper, we will use $\ell=2^7$.
 	
 	As we can see from Table~\ref{tab:singleMinded}, both Bayesian designs dominate the non-Bayesian baselines in terms of cleared instances while also being significantly better in terms of number of rounds.
 	
 	\subsection{Results for Multi-Minded Settings}
 	We now evaluate our Bayes$^\text{GEN}$ auction in settings where bidders are multi-minded. As we can observe from Table~\ref{tab:clearingMultiMinded}, our Bayesian auction outperforms both baselines in terms of clearing and rounds (on average, over the different distributions).
 	In Figure~\ref{figure:CfrAuctionRounds}, we present the distributions of results for the auction rounds using box plots. Note that our Bayesian auction always significantly outperforms SG-Auction$^\text{D}$, and it outperforms SG-Auction$^\text{I}$ for three out of the four distributions. Furthermore, note that, while SG-Auction$^\text{I}$  and SG-Auction$^\text{D}$  present very heterogeneous behavior across different distributions, our Bayesian design is much more consistent, with a third quartile that is always below 25 rounds.

 	\section{Conclusion}\label{sec:conclusion}
 	In this paper, we have presented a Bayesian iterative CA for general classes of bidder valuations. Our framework allows the auctioneer to make use of any prior information and any model of bidder values to propose new ask prices at each round of the auction. At the core of our auction is a practical Monte Carlo EM algorithm to compute the most likely clearing prices based on the bidders' revealed information. Our auction design is competitive against the design proposed by~\citet{brero2018Bayes} over single-minded valuations, for which the latter was specially designed. For general valuations, our auction design (without any special parameter tuning) outperforms a very powerful subgradient auction design with carefully tuned price increments.
 	
 	Our work gives rise to multiple promising research directions that can leverage and build on our framework. The most immediate next step is to investigate different valuation models within the belief update component. In this paper, we considered Gaussian models to better compare against prior work, but the analytic convenience of these models is no longer needed; for instance, it may be the case that different kinds of models work best for the different CATS distributions, and such insights could give guidance for real-world modeling. Another intriguing direction is to handle incentives within the framework itself, rather than rely on a separate phase with VCG or core pricing. Future work could investigate whether the auction's modeling component can gather enough information to also compute (likely) VCG payments.
 	
	\section*{Acknowledgments}
	Part of this research was supported by the SNSF (Swiss National Science
	Foundation) under grant \#156836. 	
 	
 	\bibliographystyle{aaai}
 	\bibliography{aaai19}

\begin{thebibliography}{}

\bibitem[\protect\citeauthoryear{Ausubel and
  Baranov}{2017}]{ausubel2017practical}
Ausubel, L.~M., and Baranov, O.
\newblock 2017.
\newblock {A Practical Guide to the Combinatorial Clock Auction}.
\newblock {\em The Economic Journal} 127(605).

\bibitem[\protect\citeauthoryear{Ausubel \bgroup et al\mbox.\egroup
  }{2006}]{ausubel2006clock}
Ausubel, L.~M.; Cramton, P.; Milgrom, P.; et~al.
\newblock 2006.
\newblock {The Clock-proxy Auction: A Practical Combinatorial Auction Design}.
\newblock In Cramton, P.; Shoham, Y.; and Steinberg, R., eds., {\em
  Combinatorial Auctions}. MIT Press.
\newblock chapter~5.

\bibitem[\protect\citeauthoryear{Bichler, Hao, and
  Adomavicius}{2017}]{bichler2017coalition}
Bichler, M.; Hao, Z.; and Adomavicius, G.
\newblock 2017.
\newblock {Coalition-Based Pricing in Ascending Combinatorial Auctions}.
\newblock {\em Information Systems Research} 28(1):159--179.

\bibitem[\protect\citeauthoryear{Bikhchandani and
  Ostroy}{2002}]{bikhchandani2002package}
Bikhchandani, S., and Ostroy, J.~M.
\newblock 2002.
\newblock {The Package Assignment Model}.
\newblock {\em Journal of Economic theory} 107(2):377--406.

\bibitem[\protect\citeauthoryear{Bikhchandani \bgroup et al\mbox.\egroup
  }{2001}]{bikhchandani2001linear}
Bikhchandani, S.; de~Vries, S.; Schummer, J.; and Vohra, R.~V.
\newblock 2001.
\newblock {Linear Programming and Vickrey Auctions}.
\newblock {\em IMA Volumes in Mathematics and its Applications} 127:75--116.

\bibitem[\protect\citeauthoryear{Blum \bgroup et al\mbox.\egroup
  }{2004}]{blum2004preference}
Blum, A.; Jackson, J.; Sandholm, T.; and Zinkevich, M.
\newblock 2004.
\newblock {Preference Elicitation and Query Learning}.
\newblock {\em Journal of Machine Learning Research} 5:649--667.

\bibitem[\protect\citeauthoryear{Brero and Lahaie}{2018}]{brero2018Bayes}
Brero, G., and Lahaie, S.
\newblock 2018.
\newblock {A Bayesian Clearing Mechanism for Combinatorial Auctions}.
\newblock In {\em {Proceedings of the 32nd AAAI Conference of Artificial
  Intelligence}},  941--948.

\bibitem[\protect\citeauthoryear{Brero, Lubin, and
  Seuken}{2017}]{brero2017probably}
Brero, G.; Lubin, B.; and Seuken, S.
\newblock 2017.
\newblock {Probably Approximately Efficient Combinatorial Auctions via Machine
  Learning.}
\newblock In {\em {Proceedings of the 31st AAAI Conference of Artificial
  Intelligence}},  397--405.

\bibitem[\protect\citeauthoryear{Brero, Lubin, and
  Seuken}{2018}]{brero2018MLelicit}
Brero, G.; Lubin, B.; and Seuken, S.
\newblock 2018.
\newblock {Combinatorial Auctions via Machine Learning-based Preference
  Elicitation.}
\newblock In {\em {Proceedings of the 27th International Joint Conference on
  Artificial Intelligence}},  128--136.

\bibitem[\protect\citeauthoryear{Cramton}{2013}]{Cramton2013SpectrumAuctionDesign}
Cramton, P.
\newblock 2013.
\newblock {Spectrum Auction Design}.
\newblock {\em Review of Industrial Organization} 42(2):161--190.

\bibitem[\protect\citeauthoryear{Goetzendorf \bgroup et al\mbox.\egroup
  }{2015}]{Goetzendorf:2015}
Goetzendorf, A.; Bichler, M.; Shabalin, P.; and Day, R.~W.
\newblock 2015.
\newblock {Compact Bid Languages and Core Pricing in Large Multi-item
  Auctions}.
\newblock {\em Management Science} 61(7):1684--1703.

\bibitem[\protect\citeauthoryear{Gul and Stacchetti}{2000}]{gul2000english}
Gul, F., and Stacchetti, E.
\newblock 2000.
\newblock {The English Auction with Differentiated Commodities}.
\newblock {\em Journal of Economic theory} 92(1):66--95.

\bibitem[\protect\citeauthoryear{Kwasnica \bgroup et al\mbox.\egroup
  }{2005}]{kwasnica2005new}
Kwasnica, A.~M.; Ledyard, J.~O.; Porter, D.; and DeMartini, C.
\newblock 2005.
\newblock {A New and Improved Design for Multiobject Iterative Auctions}.
\newblock {\em Management Science} 51(3):419--434.

\bibitem[\protect\citeauthoryear{Lahaie and Parkes}{2004}]{lahaie2004applying}
Lahaie, S., and Parkes, D.~C.
\newblock 2004.
\newblock {Applying Learning Algorithms to Preference Elicitation}.
\newblock In {\em Proceedings of the 5th ACM Conference on Electronic
  Commerce},  180--188.

\bibitem[\protect\citeauthoryear{Leyton-Brown, Pearson, and
  Shoham}{2000}]{leyton2000towards}
Leyton-Brown, K.; Pearson, M.; and Shoham, Y.
\newblock 2000.
\newblock {Towards a Universal Test Suite for Combinatorial Auction
  Algorithms}.
\newblock In {\em Proceedings of the 2nd ACM Conference on Electronic
  Commerce},  66--76.
\newblock ACM.

\bibitem[\protect\citeauthoryear{Milgrom and
  Segal}{2013}]{milgrom2013designingUS}
Milgrom, P., and Segal, I.
\newblock 2013.
\newblock {Designing the US Incentive Auction}.

\bibitem[\protect\citeauthoryear{Minka}{2001}]{minka2001family}
Minka, T.~P.
\newblock 2001.
\newblock {\em {A Family of Algorithms for Approximate Bayesian Inference}}.
\newblock Ph.D. Dissertation, Massachusetts Institute of Technology.

\bibitem[\protect\citeauthoryear{Nguyen and
  Sandholm}{2014}]{nguyen2014optimizing}
Nguyen, T.-D., and Sandholm, T.
\newblock 2014.
\newblock {Optimizing Prices in Descending Clock Auctions}.
\newblock In {\em {Proceedings of the 15th ACM Conference on Economics and
  Computation}},  93--110.
\newblock ACM.

\bibitem[\protect\citeauthoryear{Nguyen and Sandholm}{2016}]{nguyen2016multi}
Nguyen, T.-D., and Sandholm, T.
\newblock 2016.
\newblock {Multi-Option Descending Clock Auction}.
\newblock In {\em {Proceedings of the 2016 International Conference on
  Autonomous Agents and Multiagent Systems}},  1461--1462.
\newblock International Foundation for Autonomous Agents and Multiagent
  Systems.

\bibitem[\protect\citeauthoryear{Opper and Winther}{1998}]{opper1998bayesian}
Opper, M., and Winther, O.
\newblock 1998.
\newblock {A Bayesian Approach to Online Learning}.
\newblock {\em Online Learning in Neural Networks}  363--378.

\bibitem[\protect\citeauthoryear{Parkes}{1999}]{parkes1999bundle}
Parkes, D.~C.
\newblock 1999.
\newblock {iBundle: an Efficient Ascending Price Bundle Auction}.
\newblock In {\em {Proceedings of the 1st ACM Conference on Electronic
  Commerce}},  148--157.
\newblock ACM.

\bibitem[\protect\citeauthoryear{Sandholm}{2013}]{Sandholm2013}
Sandholm, T.
\newblock 2013.
\newblock {Very-Large-Scale Generalized Combinatorial Multi-Attribute Auctions:
  Lessons from Conducting \$60 Billion of Sourcing}.
\newblock In Vulkan, N.; Roth, A.~E.; and Neeman, Z., eds., {\em The Handbook
  of Market Design}. Oxford University Press.
\newblock chapter~1.

\bibitem[\protect\citeauthoryear{Sollich}{2002}]{sollich2002bayesian}
Sollich, P.
\newblock 2002.
\newblock {Bayesian Methods for Support Vector Machines: Evidence and
  Predictive Class Probabilities}.
\newblock {\em Machine learning} 46(1-3):21--52.

\bibitem[\protect\citeauthoryear{Train}{2009}]{train2009discrete}
Train, K.~E.
\newblock 2009.
\newblock {\em {Discrete Choice Methods with Simulation}}.
\newblock Cambridge University Press.

\bibitem[\protect\citeauthoryear{Wei and Tanner}{1990}]{wei1990monte}
Wei, G.~C., and Tanner, M.~A.
\newblock 1990.
\newblock {A Monte Carlo Implementation of the EM Algorithm and the Poor Man's
  Data Augmentation Algorithms}.
\newblock {\em Journal of the American Statistical Association}
  85(411):699--704.

\bibitem[\protect\citeauthoryear{Williams and
  Rasmussen}{2006}]{williams2006gaussian}
Williams, C.~K., and Rasmussen, C.~E.
\newblock 2006.
\newblock {\em Gaussian Processes for Machine Learning}.
\newblock The MIT Press.

\end{thebibliography}
 	
 \end{document}